\documentclass[11pt]{article}
\usepackage{fullpage}
\usepackage{subfig}
\usepackage[usenames,dvipsnames]{xcolor}
\usepackage[colorlinks,citecolor=blue,linkcolor=BrickRed]{hyperref}
	\usepackage{latexsym,graphicx,epsfig,color}
\usepackage{amsfonts,amssymb,amsmath,amsthm,amstext}
\usepackage{libertine}
\usepackage[libertine]{newtxmath}
\usepackage{enumitem}
\setitemize{itemsep=2pt,topsep=0pt,parsep=0pt}
\usepackage{url,setspace}
\usepackage{multirow}
\usepackage[compact]{titlesec}
\usepackage{rotating}
\usepackage{makeidx}
\usepackage{accents}
\usepackage{xspace}
\usepackage{cleveref}
\usepackage{algorithm,algpseudocode}
\usepackage{bm}
\usepackage{changepage} 	
\usepackage{thmtools,thm-restate}

\newcommand{\IGNORE}[1]{}
\allowdisplaybreaks

\makeindex

\newtheorem{theorem}{Theorem}[section]
\newtheorem{claim}[theorem]{Claim}

\newtheorem{lemma}[theorem]{Lemma}

\theoremstyle{definition}

\newtheorem{remark}[theorem]{Remark}

\newtheorem{defn}[theorem]{Definition}

\usepackage{thmtools,thm-restate} 

\newcommand{\E}{\mathbb{E}}

\newcommand{\pop}{\mathcal{P}}
\newcommand{\cop}{\mathcal{C}}

\newcommand{\rop}{\mathcal{R}}

\def\bw {{\bf w}}

\def \integers {\mathbb{Z}}

\newcommand{\poly}{\operatorname{poly}}

\newcommand{\danger}{\textsf{Danger}}

\newcounter{note}[section]

\newcommand{\bd}{\ensuremath{\mathbf{d}}}
\newcommand{\bdt}{\ensuremath{\bd^{t}}}

\newcommand{\vol}{\ensuremath{\textsf{vol}}}

\newcommand{\OPBP}{$(1+\beta)$-process\xspace}
\newcommand{\algo}{\ensuremath{\mathbf{DivideAndGreedy}}\xspace}

\newcommand{\initOneLiners}{%
    \setlength{\itemsep}{0pt}
    \setlength{\parsep }{0pt}
    \setlength{\topsep }{0pt}
}
\newenvironment{OneLiners}[1][\ensuremath{\bullet}]
    {\begin{list}
        {#1}
        {\initOneLiners}}
    {\end{list}}

\title{Online Carpooling using
  Expander Decompositions\footnote{This research was done under the auspices of the Indo-US Virtual Networked Joint Center IUSSTF/JC-017/2017.}}

	\author{Anupam Gupta\thanks{
	(anupamg@cs.cmu.edu)
	Computer Science Department,
	Carnegie Mellon University.
	Research supported in part by NSF awards CCF-1907820, CCF1955785, and CCF-2006953.
	}
	\and Ravishankar Krishnaswamy\thanks{
	(rakri@microsoft.com)
	Microsoft Research.
	}
		\and Amit Kumar\thanks{
		(amitk@cse.iitd.ac.in)
		Department of Computer Science and Engineering,
		Indian Institute of Technology, Delhi.
		}
	\and Sahil Singla\thanks{
        (singla@cs.princeton.edu)
        Department of Computer Science at        Princeton University
        and School of Mathematics at Institute for Advanced Study.
        Research supported in part by the Schmidt Foundation.
        }
}

\date{ \today}

\begin{document}
\maketitle

\setlength{\abovedisplayskip}{2pt}
\setlength{\belowdisplayskip}{2pt}

\begin{abstract}{
We consider the online carpooling problem: given $n$ vertices, a sequence of edges arrive over time. When an edge $e_t = (u_t, v_t)$ arrives at time step $t$, the algorithm must orient the edge either as $v_t \rightarrow u_t$ or $u_t \rightarrow v_t$, with the objective of minimizing the maximum discrepancy of any vertex, i.e., the absolute difference between its in-degree and out-degree. Edges correspond to pairs of persons wanting to ride together, and orienting denotes designating the driver. The discrepancy objective then corresponds to every person driving close to their fair share of rides they participate in.

\medskip

In this paper, we design efficient algorithms which can maintain polylog$(n,T)$ maximum discrepancy (w.h.p) over any sequence of $T$ arrivals, when the arriving edges are sampled independently and uniformly from any given graph $G$. This provides the first polylogarithmic bounds for the online (stochastic) carpooling problem. Prior to this work, the best known bounds were $O(\sqrt{n \log n})$-discrepancy for any adversarial sequence of arrivals, or $O(\log\!\log n)$-discrepancy bounds for the stochastic arrivals when $G$ is the complete graph.  

\medskip

The technical crux of our paper is in showing that the simple greedy algorithm, which has provably good discrepancy bounds when the arriving edges are drawn uniformly at random from the complete graph, also has polylog discrepancy when $G$ is an expander graph. We then combine this with known expander-decomposition results to design our overall algorithm.
}\end{abstract}




\section{Introduction}
\label{sec:intro}

Consider the following \emph{edge orientation} problem: we are given a set
$V$ of $n$ nodes, and undirected edges arrive online one-by-one. Upon
arrival of an edge $\{u,v\}$, it has to be oriented as either $u \to v$
or $v \to u$, immediately and irrevocably. The goal is to minimize the
\emph{discrepancy} of this orientation at any time $t\in [T]$ during the arrival
process, i.e., the maximum imbalance
between the in-degree and out-degree of any node.  Formally, if we let 
$\bm{\chi}^t$ to denote the orientation at time $t$ and $\delta_t^-(v)$ (resp. $\delta_t^+(v)$) to denote the number of in-edges (resp. out-edges) incident to $v$ in $\bm{\chi}^t$, then we want to minimize
\[ \max_t \text{disc}(\bm{\chi}^t) := \max_t \max_v | \delta_t^-(v) - \delta_t^+(v) |. \] 
If the entire sequence of edges is known up-front, one can use a
simple cycle-and-path-peeling argument to show that any set of
edges admit a discrepancy of at most $1$. The main focus of this work is in understanding how much loss is caused by the presence of uncertainty, since we don't have knowledge of future arrivals when we irrevocably orient an edge.

This problem was proposed by Ajtai et al.~\cite{AANRSW-Journal98} as a special
case of the \emph{carpooling problem} where hyperedges arrive online,
each representing a carpool where one person must be designated as a
driver. The ``fair share'' of driving for person $i$ can be defined as
$\sum_{e: i \in e} 1/|e|$, and we would like each person to drive
approximately this many times. In the case of graphs where each
carpool is of size $|e| = 2$, this carpooling problem is easily
transformed into the edge-orientation problem.

Ajtai et al.\ showed that while deterministic algorithms cannot have
an $o(n)$ discrepancy, they gave a randomized ``local
greedy'' which has an expected discrepancy (for any $T \geq 1$) of $O(\sqrt{n \log n})$ for any online input sequence of $T$ arrivals.  Indeed, note that the discrepancy bound is independent of the length of the sequence $T$, and depends only
on the number of nodes, thus giving a non-trivial improvement over the naive random assignment, which will incur a discrepancy of $O(\sqrt{T \log n})$. Intriguingly, the lower bound they show for online algorithms is only $\Omega((\log n)^{1/3})$---leaving a large gap between the upper and
lower bounds. 

Given its apparent difficulty in the adversarial online model, Ajtai et al.\ proposed a stochastic model, where each edge is an
independent draw from some underlying probability distribution over
pairs of vertices. They considered the the uniform distribution, which
is the same as presenting a uniformly random edge of the complete
graph at each time. In this special case, they showed that the greedy
algorithm (which orients each edge towards the endpoint with lower
in-degree minus out-degree) has expected discrepancy
$\Theta(\log\!\log n)$. Their analysis crucially relies on the
structure and symmetry of the complete graph.

In this paper, we consider this stochastic version of the problem for general graphs:
i.e., given an arbitrary simple graph $G$, the online input is a sequence of
edges chosen independently and uniformly at random (with replacement) from the edges of
this graph $G$\footnote{It is possible to extend our results, by losing a $\log T$ factor, to edge-weighted distributions where an edge is drawn i.i.d. with probability proportional to its weight. Since this extension
uses standard ideas like bucketing edges with similar weights, we restrict our attention    to  arrivals from a graph $G$ for simplicity.}. 
  Our main result is the following:

\begin{theorem}[Main Theorem] \label{thm:final}
  There is an efficient algorithm for the edge-orientation problem that
  maintains, w.h.p, a maximum discrepancy of $O(\poly\log (nT))$  on input sequences formed by i.i.d.\ draws from the edges of a given graph $G$.
\end{theorem}



\subsection{Our Techniques}


Let us fix some notation. Given a (multi)graph $G = (V,E)$ with
$|V| = n$, the algorithm is presented with a vector $v^t$ at each time
as follows. A uniformly random edge $(u,v) \in G$ is sampled, and the
associated characteristic vector $v^t = \mathbf{e}_u - \mathbf{e}_v$
is presented to the algorithm, where $\mathbf{e}_u \in \mathbb{R}^n$
has all zeros except index $u$ being $1$. The algorithm must
immediately sign $v^t$ with $\chi^t \in \{-1, 1\}$, to keep
the discrepancy bounded at all times $t$. Here the discrepancy
of node $u$ at time $t$ is the $u^{th}$ entry of the vector
$\sum_{s \leq t} \chi^{s} v^{s}$ (which could be negative), and the discrepancy of the
algorithm is the maximum absolute discrepancy over all vertices, i.e., 
$ \Big\| \sum_{s \leq t} \chi^{s} v^{s} \Big\|_{\infty}$ .

A natural algorithm
is to pick a uniformly random orientation for each arriving edge. This
maintains zero expected discrepancy at each node. However, the large
variance may cause the maximum discrepancy over nodes to be as large
as $\Omega(\sqrt{T})$, where $T$ the total number of edges (which is
the same as the number of time-steps). For example, this
happens even on $T$ parallel edges between two nodes. In this case, however, the \emph{greedy
algorithm} which orients the edge from the vertex of larger discrepancy to that of smaller discrepancy works well. Indeed it is not known to be
bad for stochastic instances. (Since it is a
deterministic algorithm, it can perform poorly on adversarial inputs due to known $o(n)$ lower bounds~\cite{AANRSW-Journal98}.)

Building on the work of Ajtai et al. who consider stochastic arrivals on complete graphs, the first step towards our overall algorithm is to consider the problem on \emph{expander graphs}.  At a
high level, one hurdle to achieving low discrepancy in the stochastic
case is that we reach states where both endpoints of a randomly-chosen edge
already have equally high discrepancy. Then, no matter how we orient the edge,
we increase the maximum discrepancy. But this should not happen in
expander graphs: if $S$ is the set of ``high'' discrepancy vertices,
then the expansion of the graph implies that $|\partial S|$ must be a
large fraction of the total number of edges incident to
$S$. Therefore, intuitively, we have a good chance of reducing the
discrepancy if we get edges that go from $S$ to low-degree nodes. 
To make this idea formal, we relate the greedy process on expander graphs
$G$ to the so-called \OPBP over an \emph{easier arrival} sequence where the end-points of a new edge are chosen from a \emph{product distribution}, where the probability of choosing a vertex is proportional to its
degree in $G$. However, in the \OPBP\footnote{The name \OPBP stems from the notion for an analogous load-balancing (or) balls-and-bins setting~\cite{PTW-Journal15}, this process would be like the $(1+\beta)$-fractional version of the power-of-two choices process.}, the algorithm orients a new edge greedily with only probability $\beta$ for some small value of $\beta$, and does a random orientation with the remaining probability $(1-\beta)$.

Indeed, we compare these two processes by showing that (a)~the expected increase of a
natural potential $\Phi := \sum_v \cosh(\lambda \, {\rm discrepancy}(v))$---which can
be thought of as a soft-max function---is lower for the greedy algorithm on expanders when compared to the \OPBP on the product distribution, and (b)~the same potential increases very slowly
(if at all) on the product distribution.  A similar idea was used by Peres et al.~\cite{PTW-Journal15} for a 
related stochastic load balancing problem; however, many of the technical details are different.


The second component of the algorithm is to decompose a general
graph into expanders. This uses the (by-now commonly used) idea of
expander decompositions. Loosely speaking, this says that the edges of
any graph can be decomposed into some number of smaller graphs (each
being defined on some subset of vertices), such that (a)~each of these
graphs is an expander, and (b)~each vertex appears in only a
poly-logarithmic number of these expanders. Our arguments for
expanders require certain weak-regularity properties---namely the
degrees of vertices should not be too small compared to the average
degree---and hence some care is required in obtaining decompositions
into such expanders. These details appear in \S\ref{sec:expand-decomp}.

Our overall algorithm can then be summarized in Algorithm~\ref{alg:expander-greedy}.

\begin{algorithm}
	\caption{{\algo} (graph $G=(V,E)$)}
	\label{alg:expander-greedy}
	\begin{algorithmic}[1]
		\State run the expander-decomposition algorithm in~\Cref{thm:regular-expanders} (in~\Cref{sec:tieup}) on $G$ to obtain a collection ${\cal P} = \{G_1, G_2, \ldots, G_k\}$ of edge-disjoint expander graphs.
		\State initialize ${\cal H} = \{H_1, H_2, \ldots H_k\}$ to be a collection of empty graphs, where $H_i$ is the directed multi-graph consisting of all edges which have arrived corresponding to base graph $G_i$, along with their orientations assigned by the algorithm upon arrival.
		\For {each new edge $e \equiv \{u,v\}$ that arrives at time-step $t$}  \label{alg:arrival}
		\State let $i$ denote the index such that $e \in G_i$ according to our decomposition. \label{alg:identify} 
		\State add $e$ to $H_i$, and orient $e$ in a greedy manner w.r.t $H_i$, i.e., from $u$ to $v$ if ${\rm disc}_{H_i}(u) \geq {\rm disc}_{H_i}(v)$, where ${\rm disc}_H(w) = \delta_{H_i}^{{\rm in}}(w) - \delta_{H_i}^{{\rm out}}(w)$ is the in-degree minus out-degree of any vertex $w$ in the  current sub-graph $H_i$ maintained by the algorithm.  \label{alg:greedy}
		\EndFor
	\end{algorithmic}
\end{algorithm}


\subsection{Related Work}
\label{sec:related-work}


The study of discrepancy problems has a long history; see the 
books~\cite{Matousek-Book09,Chazelle-Book01} for details on the
classical work. 
The problem of online discrepancy minimization was studied by
Spencer~\cite{Spencer77}, who showed an $\Omega(\sqrt{T})$ lower bound
for for adaptive adversarial arrivals. More refined lower bounds were given
by B\'ar\'any~\cite{Barany79};
see~\cite{BJSS-STOC20} for many other references. 
Much more recently, Bansal and
Spencer~\cite{BS-arXiv19} and Bansal et al.~\cite{BJSS-STOC20}
consider a more general vector-balancing problem, where each request
is a vector $v^t \in \mathbb{R}^n$ with $\|v^t \|_{\infty} \leq1$, and
the goal is to assign a sign $\chi^t \in \{-1,1\}$ to each vector to
minimize $\| \sum_t \chi^t v^t\|_\infty$, i.e., the largest coordinate of
the signed sum. Imagining each edge $e_t = \{u,v\}$ to be the vector
$\frac{1}{\sqrt{2}}(\mathbf{e}_u - \mathbf{e}_v)$ (where this initial
sign is chosen arbitrarily) captures the edge-orientation problem up
to constant factors. Bansal et al.\ gave an
$O(n^2 \log (nT))$-discrepancy algorithm for the natural stochastic
version of the problem under general distributions. 
For some special geometric problems, they gave an algorithm that
maintains $\poly(s, \log T, \log n)$ discrepancy for sparse vectors that have only $s$ non-zero coordinates. These improve on the work of Jiang et
al.~\cite{JKS-arXiv19}, who give a sub-polynomial discrepancy coloring
for online arrivals of points on a line. A related variant of these 
geometric problems
was also studied in Dwivedi et al.~\cite{DFGR19}.


Very recently, an independent and exciting work of Alweiss, Liu, and
Sawhney~\cite{ALS-arXiv20} gave a randomized algorithm that maintains
a discrepancy of $O(\log (nT)/\delta)$ for any input sequence chosen
by an oblivious adversary with probability $1-\delta$, even for the
more general vector-balancing problem for vectors of unit Euclidean
norm (the so-called K\'oml\'os setting). Instead of a potential based
analysis like ours, they directly argue  why a carefully chosen 
randomized greedy algorithm ensures  w.h.p. that the
discrepancy vector is  always sub-Gaussian. A concurrent work of Bansal et al.~\cite{BJMSS-arXiv20} also obtains similar results for i.i.d. arrivals, but they use a very different potential than our expander-decomposition approach.   It is an interesting open question to extend our approach to hypergraphs and re-derive their results.


\subsection{Notation} \label{sec:notation}

We now define some graph-theoretic terms that are useful for the remainder of the paper.

\begin{defn}[Volume and $\alpha$-expansion]
Given any graph $G=(V,E)$, and set $S \subseteq V$ its \emph{volume} is
defined to be $\vol(S) := \sum_{v \in S} \text{degree}(v)$. We  say $G$ is an \emph{$\alpha$-expander} if  
\[ \min_{S\subseteq V} \frac{|E(S, V \setminus S)|}{\min\{ \vol(S), \vol(V \setminus S)  \}} \geq \alpha.
\]
\end{defn}

We will also need the following definition of ``weakly-regular''
graphs, which are graphs where every vertex has degree \emph{at least} a constant factor of the average degree. Note that the \emph{maximum} degree can be arbitrarily larger than the average degree.

\begin{defn}[$\gamma$-weakly-regular] For $\gamma \in [0,1]$, a graph $G=(V,E)$ is called \emph{$\gamma$-weakly-regular} if every vertex $v\in V$ has degree at least $\gamma \cdot {\sum_{u\in V} \text{degree}(u)}/{|V|}$.
\end{defn}

\begin{defn}[Discrepancy Vector]
Given any directed graph $H=(V,A)$ (representing all the oriented edges until any particular time-step), let $\bd \in \integers^{|V|}$ represent the discrepancy vector of the current graph, i.e. the $v^{{\rm th}}$ entry of $\bd$, denoted by $d_v$ is the difference between the number of in-edges incident at $v$ and the number of out-endges incident at $v$ in $H$.
\end{defn}

\section{The Greedy Algorithm on Expander Graphs} 

In this section, we consider the special case when the
graph $G$ is an expander. More formally, we show that the greedy algorithm is actually good for such graphs.

\begin{defn}[Expander Greedy Process]
The greedy algorithm maintains a current discrepancy $d^t_v$ for each vertex $v$, which is the in-degree minus out-degree of every vertex among the previously arrived edges. Initially,  $d^1_v = 0$ for every vertex $v$ at the beginning of time-step $1$. At each time~$t \geq 1$, a uniformly random edge $e \in G$ with end-points $\{u,v\}$ is presented to the algorithm, and suppose w.l.o.g. $d^t_u \geq d^t_v$, i.e., $u$ has larger discrepancy (ties broken arbitrarily). Then, the algorithm orients the edge from $u$ to $v$. The discrepancies of $u$ and $v$ become $d^{t+1}_u = d^{t}_u -1$ and $d^{t+1}_v = d^{t}_u + 1$, and other vertices' discrepancies are unchanged.
\end{defn}

\begin{theorem}
  \label{thm:main}
    Consider any $\gamma$-weakly-regular $\alpha$-expander $G$, and suppose edges are arriving as independent samples from $G$ over a horizon of $T$ time-steps. Then, the greedy algorithm maintains a discrepancy $d^t_v$ of $O(\log^5 nT)$ for
    every time $t$ in $[0\ldots T]$ and every vertex $v$, as long as
    $\alpha \geq 6\lambda$, $\gamma \geq \lambda^{1/4}$, where $\lambda = O(\log^{-4} nT)$. 
\end{theorem}


For the sake of concreteness, it might be instructive to assume $\alpha \approx \gamma \approx O(\frac{1}{\log n})$, which is roughly what we will obtain from our expander-decomposition process. 

\subsection{Setting Up The Proof} \label{sec:setup}

Our main idea is to introduce \emph{another random process} called the \OPBP, and show that the \OPBP stochastically dominates the expander-greedy process in a certain manner, and separately bound the behaviour of the \OPBP subsequently. By combining these two, we get our overall analysis of the expander-greedy process.

To this end, we first define a random arrival sequence where the end-points of each new edge are actually sampled independently from a \emph{product distribution}.

\begin{defn}[Product Distribution]
Given a set $V$ of vertices  with associated weights
  $\{w_v \geq 0 \mid v \in V\}$, at each time $t$, we select two vertices $u,v$ as {\em two independent samples} from $V$, according to the distribution where  any vertex $v \in V$
  is chosen with probability $\frac{w_{v}}{\sum_{v' \in V} w_{v'}}$,
  and  the vector $v^t := \chi_u - \chi_v$ is presented to the
  algorithm.
\end{defn}

We next define the \OPBP, which will be crucial for the analysis.
\begin{defn}[\OPBP on product distributions]
Consider a product distribution over a set of vertices $V$. When presented with a vector $v^t := \chi_u - \chi_v$ from this product distribution at time $t$, the \OPBP assigns a sign to the vector $v^t$ as follows: 
with
   probability $(1-\beta)$, it assigns it uniformly $\pm 1$, and only
  with the remaining probability $\beta$ it uses the greedy algorithm
   to sign this vector.
\end{defn}
Note that setting $\beta=1$ gives us back the greedy algorithm, and
$\beta=0$ gives an algorithm that assigns a random sign to each vector.

\begin{remark}
The original \OPBP was in fact introduced in~\cite{PTW-Journal15}, where Peres et al. analyzed a general load-balancing process over $n$ bins (corresponding to vertices), and balls arrive sequentially. Upon each arrival, the algorithm gets to sample a random edge from a $k$-regular expander\footnote{Actually their proof works for a slightly more general notion of expanders, but which is still insufficient for our purpose.} $G$ over the bins, and places the ball in the lighter loaded bin among the two end-points of the edge. They show that this process maintains a small maximum load, by relating it to an analogous \OPBP, where instead of sampling an edge from $G$, \emph{two bins} are chosen uniformly at random, and the algorithm places the ball into a random bin with probability $1-\beta$, and the lesser loaded bin with probability $\beta$. Note that their analysis inherently assumed that the two vertices are sampled from
the uniform distribution where all weights $w_u$ are
equal. By considering arbitrary product
distributions, we are able to handle arbitrary graphs with
a non-trivial conductance, i.e., even those that \emph{do not} satisfy the $k$-regularity property. 
This is crucial for us because the expander decomposition algorithms, which reduce general graphs to a
collection of expanders, do not output regular expanders.
\end{remark}

Our analysis will also involve a potential function (intuitively the soft-max of the vertex discrepancies) for both the expander-greedy process as well as the \OPBP. 
\begin{defn}[Potential Function]
Given vertex discrepancies $\bd \in \integers^{|V|}$, define
\begin{gather}
\Phi(\bd) := \sum_v \cosh(\lambda d_v), \label{eq:pot}
\end{gather}
where $\lambda < 1$ is a suitable parameter to be optimized. 
\end{defn}
Following many prior works, we use the hyperbolic cosine function to symmetrize for positive and negative discrepancy values. When $\bd$ is clear from the context, we will write $\Phi(\bd)$ as $\Phi$. We will also use $\bdt$ to refer to the discrepancy vector at time $t$, and $d^t_u$ to the discrepancy of $u$ at time $t$.  We will often ignore the superscript $t$ if it is clear from the context. 


We are now ready to define the appropriate parameters of the \OPBP. Indeed, given the expander-greedy process defined on graph $G$, we construct an associated \OPBP where for each vertex $v$, the probability of sampling any vertex in the product distribution is proportional to its degree in $G$, i.e., $w_v = {\rm degree}_G(v)$ for all $v \in V$. We also set the $\beta$ parameter equal to $\alpha$, the conductance of the graph $G$.

\subsection{One-Step Change in Potential} \label{sec:overview}


The main idea of the proof is to  use a  \emph{majorization argument} to argue that \emph{the expected one-step change} in potential
of the expander process can be upper bounded by that of the \OPBP, if the two processes start at the same discrepancy configuration $\bdt$. 
Subsequently, we bound the one-step change for the \OPBP in \cref{sec:beta-process}.


To this end, consider a time-step $t$, where the current discrepancy vector of the expander process is $\bdt$. 
Suppose the next edge in the expander process is $(i,j)$, where $d^t_i >
d^t_j$. Then  the greedy algorithm will always
choose a sign such that $d_i$ decreases by $1$, and $d_j$ increases by
$1$. Indeed, this ensures the overall potential is non-increasing
unless $d_i = d_j$. More importantly, the potential term for other
vertices remains unchanged, and so we can express the expected change
in potential as having contributions from precisely two terms, one due
to $d_i \to d_i - 1$ (called the \emph{decrease term}), and denoted as
$\Delta_{-1}(t)$, and one due to $d_j \to d_j + 1$ (the \emph{increase term}), denoted as $\Delta_{+1}(t)$: 
\begin{align}
\E_{(i,j) \sim G}[\Delta \Phi] &= \E_{(i,j) \sim G}\Big[\Phi(\bd^{t+1}) - \Phi(\bdt) \Big] \notag \\
& \hspace{-1cm} = \underbrace{\E_{(i,j})\Big[\cosh(\lambda (d_i-1)) -
                                 \cosh(\lambda (d_i)) \Big]}_{=: \Delta_{-1} (\bdt)} + \underbrace{\E_{(i,j)}
                                 \Big[\cosh(\lambda (d_j+1)) -
                                 \cosh(\lambda (d_j))\Big]}_{=: \Delta_{+1} (\bdt)} . \notag 
\end{align}

Now, consider the \OPBP on the vertex set $V$, where the product
distribution is given by  weights $w_u = \deg(u)$ for each $u \in
V$, starting with the same discrepancy
vector $\bdt$ as the expander process at time $t$. Then, if
$u$ and $v$ are the two vertices sampled independently according to
the product distribution, then by its definition, the \OPBP signs this pair randomly with probability $(1-\beta)$, and greedily with probability $\beta$. 
For the sake of analysis, we define two terms analogous to $\Delta_{-1}
(\bdt)$ and $\Delta_{+1} (\bdt)$ for the \OPBP. To
this end, let $i \in \{u,v\}$ denote the identity of the random vertex
to which the \OPBP assigns $+1$. Define
\begin{gather}
  \widetilde{\Delta}_{+1} (\bdt) : =\E_{(u,v) \sim \bw \times \bw}
  \Big[\cosh(\lambda (d_i+1)) - \cosh(\lambda (d_i))\Big], \label{eq:2}
\end{gather}
where $\bw \times \bw$ refers to two independent choices from the product distribution corresponding to $w$. 
Similarly let $j \in \{u,v\}$ denote the identity of the random vertex
to which the \OPBP assigns $-1$, and define
\begin{gather}
  \widetilde{\Delta}_{-1}(\bdt) :=
  \E_{(u,v) \sim \bw \times \bw} \Big[\cosh(\lambda (d_j-1)) -
  \cosh(\lambda (d_j))\Big]. \label{eq:3}
\end{gather}

In what follows, we bound $\Delta_{-1} (\bdt) \leq \widetilde{\Delta}_{-1} (\bdt)$ through a coupling argument, and similarly bound $\Delta_{+1} (\bdt) \leq \widetilde{\Delta}_{+1} (\bdt)$ using a separate coupling.

A subtlety: the expected one-step change in $\Phi$ in the
expander process precisely equals
$\Delta_{-1} (\bdt) + \Delta_{+1}(\bdt)$. However, if we define an
analogous potential for the \OPBP, then the one-step change in
potential there \emph{does not} equal the sum
$\widetilde{\Delta}_{-1} (\bdt) + \widetilde{\Delta}_{+1}
(\bdt)$. Indeed, we sample $u$ and $v$ i.i.d.\ in the \OPBP,
it is possible that $u = v$ and therefore the one-step change in
potential is $0$, while the sum
$\widetilde{\Delta}_{-1} (\bdt) + \widetilde{\Delta}_{+1} (\bdt)$ will
be non-zero. Hence the following lemma does not bound the expected
potential change for the expander process by that for the \OPBP (both
starting from the same state), but by this surrogate
$\widetilde{\Delta}_{-1} (\bdt) + \widetilde{\Delta}_{+1} (\bdt)$, and it is this
 surrogate sum that we bound in~\Cref{sec:beta-process}.

\subsection{The Coupling Argument} \label{sec:coupling}

We now show a coupling between the expander-greedy process and the \OPBP defined in ~\Cref{sec:setup}, to bound the expected one-step change in potential for the expander process. 
\begin{lemma}
  \label{lem:coupling}
Given an $\alpha$-expander $G = (V,E)$, let $\bdt \equiv (d_v \, : v
\in V)$ denote the current discrepancies of the vertices at any time
step $t$ for the expander-greedy process. Consider a hypothetical \OPBP on vertex set $V$ with $\beta = \alpha$, the
weight of vertex $v \in V$ set to $w_{v} = \deg(v)$, and starting from the same discrepancy state $\bdt$. Then:
\begin{OneLiners}
\item[(a)] $\Delta_{-1} (\bdt) \leq \widetilde{\Delta}_{-1} (\bdt)$,
  ~~and~~ (b) $\Delta_{+1} (\bdt) \leq \widetilde{\Delta}_{+1} (\bdt)$.
\end{OneLiners}
Hence the expected one-step change in potential $\E[ \Phi(\bd^{t+1})
- \Phi(\bdt)] \leq \widetilde{\Delta}_{-1} (\bdt) +
\widetilde{\Delta}_{+1} (\bdt)$. 
\end{lemma}

\begin{proof}
We start by renaming the vertices  in $V$ such that $d_n \leq d_{n-1} \leq \ldots \leq d_1$. Suppose the next edge   in the expander process corresponds to indices $i,j$ where $i<j$. 
	 We prove the lemma statement by two separate coupling
         arguments, which crucially depend on the following
         claim. Intuitively, this claim shows that a $-1$ is more
         likely to appear among the high discrepancy vertices of $G$ in
         the expander process than the \OPBP (thereby having a lower
         potential), and similarly a $+1$ is more likely to appear among
         the low discrepancy vertices of $G$ in the expander process
         than in the \OPBP. Peres et al.~\cite{PTW-Journal15} also prove
         a similar claim for stochastic load balancing, but they
         only consider uniform distributions.

\begin{claim} \label{cl:good-prefix}
	For any $k \in [n]$, if $S_k$ denotes the set of vertices with
        indices $k' \in [k]$ (the $k$ highest discrepancy vertices) and $T_k$ denotes $V\setminus S_k$, then
        \begin{align*}
        \Pr_{(i,j) \sim G}[-1 \in S_k] &\geq  \Pr_{(u,v) \sim \bw \times \bw}[-1
                                         \in S_k] \quad \text{and} \quad 
          \Pr_{(i,j) \sim G}[+1 \in T_k] &\geq  \Pr_{(u,v) \sim \bw \times \bw}[+1 \in T_k] \, .
        \end{align*}
       Above, we abuse notation and use the terminology `$-1 \in S_k$' to denote that the vertex whose discrepancy decreases falls in the set $S_k$ in the corresponding process.
\end{claim}
\begin{proof}
Fix an index $k$, and let $\rho:= \frac{\vol(S_k)}{\vol(V)}$ be the
relative volume of $S_k$, i.e., the fraction of edges of $G$ incident to the
$k$ nodes of highest degree.
	First we consider the \OPBP on $V$. With $(1-\beta)$, probability we assign a sign to the input vector uniformly at random. Therefore, conditioned on this choice, a vertex in $S_k$ will get a $-1$ sign with probability 
	$$ \frac{1}{2} \cdot \Pr[u \in S_k] + \frac{1}{2} \Pr[v \in
        S_k]~~ =~~ \frac{\vol(S_k)}{\vol(V)} ~~=~~ \rho,$$ where $u$ and $v$ denote the two vertices chosen by the \OPBP process. With probability $\beta$, we will use the greedy algorithm, and so $-1$ will appear on a vertex in $S_k$ iff at least one of the two chosen vertices lie in $S_k$. Putting it together, we get 
	\begin{align}  \Pr_{(u,v) \sim \bw \times \bw}[-1 \in S_k] 
          &= (1-\beta)\cdot \frac{\vol(S_k)}{\vol(V)} + \beta \cdot
          {\Pr_{(u,v) \sim \bw \times \bw} [\{u,v\} \cap S_k \neq \emptyset]} \notag \\
          &= (1-\beta) \cdot \rho + \beta \cdot
                               \left(1 - (1 - \rho)^2 \right) 	~~=~~ (1 + \beta - \beta \cdot \rho
            ) \cdot \rho. \label{eq:4}
	\end{align}
	
	Now we consider the expander process. A vertex in $S_k$ gets -1 iff the chosen edge has at least one end-point in $S_k$. Therefore, 
	\begin{align*} &\Pr_{(i,j) \sim G}[-1 \in S_k]  ~~ =~~ \Pr[i \in S_k] ~~ =~~ \frac{|E(S_k,S_k)| + |E(S_k, V\setminus S_k)|}{|E|} \\
	& \quad = \frac{\big( 2|E(S_k,S_k)| + |E(S_k, V\setminus S_k)|\big) + |E(S_k, V\setminus S_k)|}{2|E|}
	~~=~~ \frac{\vol(S_k) + |E(S_k, V\setminus S_k)|}{\vol(V)}.
	\end{align*}
	
Recalling that $\beta = \alpha$, and that $G$ is an $\alpha$-expander, we consider two cases:
	
	\textbf{Case 1}: If $\vol(S_k) \leq \vol(V\setminus S_k)$, we  use
	\begin{align*}  \Pr_{(i,j) \sim G}[-1 \in S_k]    &~~=~~ \frac{\vol(S_k) + |E(S_k, V\setminus S_k)|}{\vol(V)} \\
	&~~\geq~~  (1+\alpha)\frac{\vol(S_k)}{\vol(V)} ~~=~~ (1+\beta)\rho \geq   \Pr_{(u,v) \sim \bw \times \bw}[-1 \in S_k].
	\end{align*}
	
	\textbf{Case 2}: If $\vol(S_k) > \vol(V\setminus S_k)$, we use
	\begin{align*} 
	\Pr_{(i,j) \sim G}[-1 \in S_k]  &~~=~~ \frac{\vol(S_k) + |E(S_k, V\setminus S_k)|}{\vol(V)} ~~\geq~~  \frac{\vol(S_k) + \alpha \cdot \vol(V \setminus S_k)}{\vol(V)} \\  
	 &~~\geq~~ \Big(1 + \beta \cdot \frac{\vol(V \setminus S_k)}{\vol(V)} \Big) \cdot \rho ~~=~~   \Pr_{(i,j) \sim \bw \times \bw}[-1 \in S_k],
	\end{align*}
        where the last equality uses~(\ref{eq:4}).
	This completes the proof of $\Pr_{(i,j) \sim G}[-1 \in S_k] \geq  \Pr_{(i,j) \sim \bw}[-1 \in S_k]$. 
	One can similarly show $\Pr_{(i,j) \sim G}[+1 \in T_k] \geq  \Pr_{(u,v) \sim \bw \times \bw}[+1 \in T_k] $, which completes the proof of the claim.
\end{proof}

~\Cref{cl:good-prefix} shows that we can establish a coupling between the two
processes such that if $-1$ belongs to $S_k$ in \OPBP, then the same
happens in the expander process. In other words, there is a joint
sample space $\Omega$ such that for any outcome $\omega \in \Omega,$
if vertices $v_a$ and $v_b$ get sign $-1$ in the expander process and
the \OPBP respectively, then $a \leq b$.

Let $\bd$ and ${\widetilde \bd}$ denote the discrepancy vectors in the
expander process and the \OPBP after the -1 sign has been assigned,
respectively. Now, since both the processes start with the same discrepancy
vector $\bdt$, we see that for any fixed outcome $\omega \in \Omega,$
the vector ${\widetilde \bd}$ majorizes $\bd$ in the following sense. 

\begin{defn}[Majorization] Let ${\bf a}$ and ${\bf b}$ be two real vectors of the same length $n$. Let $\overrightarrow{{\bf a}}$ and $\overrightarrow{{\bf b}}$ denote the vectors ${\bf a}$ and ${\bf b}$ with coordinates rearranged in descending order respectively. We say that ${\bf a}$ {\em majorizes} ${\bf b}$, written ${\bf a} \succeq {\bf b}$, if for all $i$, $1 \leq i \leq n$, we have $\sum_{j=1}^i \overrightarrow{{\bf a}}_j \geq \sum_{j=1}^i \overrightarrow{{\bf b}}_j. $
\end{defn}

One of the properties of majorization~\cite{Hardy} is that any convex and symmetric function of the discrepancy vector (which $\Phi$ is)  satisfies
that $\Phi(\bd) \leq \Phi({\widetilde \bd})$. 
Thus, for any fixed outcome $\omega$, the change in potential in the expander
process is at most that of the surrogate potential in the \OPBP. Since $\Delta_{-1}(\bdt)$ and
$\widetilde{\Delta}_{-1} (\bdt)$ are just the expected change of these
quantities in the two processes (due to assignment of -1 sign), the
first statement of the lemma follows. Using an almost identical proof, we can also show the second statement. (Note that we may need to redefine the coupling
between the two processes to ensure that if vertices $v_a, v_b$ get
sign $+1$ as above, then $b \leq a$.)
\end{proof}





\subsection{Analyzing One-Step $\Delta \Phi$ of the \OPBP} \label{sec:beta-process}

Finally we bound the one-step change in (surrogate) potential of the
\OPBP starting at discrepancy vector $\bdt$; recall the definitions of $\widetilde{\Delta}_{-1}(\bdt)$ and $\widetilde{\Delta}_{+1}(\bdt)$ from~\Cref{sec:overview}.


\begin{lemma} \label{lem:beta-process}
	If $\Phi(\bdt) \leq (nT)^{10}$,
          and if the weights $w_v$ are such that for all $v$, $\frac{w_v}{\sum_{v'} w_{v'}}\geq \frac{\gamma}{n}$ (i.e.,
        the minimum weight is at least a $\gamma$ fraction of the average weight), 
        then we have that
        \[ \widetilde{\Delta}_{-1}(\bdt) +
          \widetilde{\Delta}_{+1}(\bdt) \leq O(1), \] 
    as long as $\beta \geq 6\lambda$, $\gamma \geq 16 \lambda^{1/4}$, and $\lambda = O(\log^{-4} nT)$.
\end{lemma}

\begin{proof}
	Let
        $u$ be an arbitrary vertex in $V$, and we condition on the fact that the first vertex chosen by the \OPBP is $u$. Then, we show that 
        \begin{gather*}
          \E_{v \sim \bw} \Big[\cosh(\lambda (d_i-1)) -
          \cosh(\lambda (d_i)) + \cosh(\lambda (d_j+1)) -
          \cosh(\lambda (d_j)) \, \Big| \, u \textrm{ is sampled
            first}\Big],
        \end{gather*}
        is $O(1)$ regardless of the choice of $u$, where we assume that 
        $i$ is the random vertex which is assigned $-1$ by the \OPBP,
        and $j$ is the random vertex which is assigned $+1$. The proof of the lemma then
        follows by removing the conditioning on $u$. 

Following~\cite{BS-arXiv19,BJSS-STOC20}, we use the first two terms of the Taylor expansion of $cosh(\cdot)$ to upper bound the difference terms of the form $\cosh(x+1) - \cosh(x)$ and $\cosh(x-1) - \cosh(x)$. To this end, note that, if $|\epsilon| \leq 1$ and $\lambda < 1$, we have that 
\begin{align*}
\cosh(\lambda (x + \epsilon)) - \cosh(\lambda x) & \textstyle \leq  \epsilon \lambda \sinh(\lambda x) + \frac{\epsilon^2}{2!} \lambda^2 \cosh(\lambda x) + \frac{\epsilon^3}{3!} \lambda^3 \sinh(\lambda x) + \ldots \\
& \leq \epsilon \lambda \sinh(\lambda x) + \epsilon^2 \lambda^2 \cosh(\lambda x).
\end{align*}
Using this, we proceed to bound the following quantity (by setting $\epsilon = -1$ and $1$ respectively):
\begin{align*} 
&\E_{v \sim \bw} \Big[ \underbrace{- \lambda \big( \sinh(\lambda d_i) - \sinh (\lambda d_j) \big)}_{=: -L} +  \underbrace{\lambda^2  \big(  \cosh(\lambda d_i) + \cosh (\lambda d_j) \big)}_{=:Q} \, \Big| \, u \textrm{ is sampled first}\Big] .
\end{align*}
We refer to $L=\lambda \big( \sinh(\lambda d_i) - \sinh (\lambda d_j) \big)$ and $Q=\lambda^2 \big( \cosh(\lambda (d_i)) + \cosh (\lambda d_j) \big)$ as the \emph{linear} and \emph{quadratic} terms, since they arise from the first- and second-order derivatives in the Taylor expansion.

To further simplify our exposition, we define the following random variables:
\begin{OneLiners}
\item[(i)] $u_{>}$ is the identity of the vertex among $u,v$ with higher discrepancy, and $u_{<}$ is the other vertex. Hence we have that $d_{u_>} \geq d_{u_<}$.
\item[(ii)] $G$  denotes the random variable $\lambda \big( \sinh(\lambda d_{u_>}) - \sinh (\lambda d_{u_<}) \big)$, which indicates an analogous term to $L$, but if we exclusively did a greedy signing always (recall that the greedy algorithm would always decrease the larger discrepancy, but the \OPBP follows a uniformly random signing with probability $(1-\beta)$ and follows the greedy rule only with probability $\beta$).
\end{OneLiners}

Finally, for any vertex $w \in V$, we let $\danger(w) = \{ v :  | d_w - d_v | < \frac{2}{\lambda}\}$ to denote the set of vertices with discrepancy close to that of $w$, where the gains from the term corresponding to $\beta G$ are insufficient to compensate for the increase due to $Q$.

We are now ready to proceed with the proof. Firstly, note that, since the \OPBP follows the greedy algorithm with probability $\beta$ (independent of the choice of the sampled vertices $u$ and $v$), we have that 
\begin{eqnarray}
\label{eq:usample}
\E_v[L \mid u \textrm{ is sampled first}] ~~=~~ (1-\beta) 0 + \beta \E_v[G \mid u \textrm{ is sampled first}].
\end{eqnarray}

Intuitively, the remainder of the proof proceeds as follows: suppose $d_{u_>}$ and $d_{u_<}$ are both non-negative (the intuition for the other cases are similar). Then, $Q$ is proportional to $\lambda^2 \cosh (\lambda d_{u_>})$. Now, if $d_{u_>} - d_{u_<}$ is sufficiently large, then $G$ is proportional to $\lambda \sinh (\lambda d_{u_>})$, which in turn is close to $\lambda \cosh (\lambda d_{u_>})$. As a result, we get that as long as $\lambda = O(\beta)$, the term $- \beta G +  Q$ can be bounded by $0$ for each choice of $v$ such that $d_{u_>} - d_{u_<}$ is large.

However, what happens when $d_{u_>} - d_{u_<}$ is small, i.e., when $v$ falls in $\danger(u)$? Here, the $Q$ term is proportional to $\lambda^2 \cosh (\lambda d_{u})$, but the $G$ term might be close to $0$, and so we can't argue that $- \beta G + Q \leq O(1)$ in these events. Hence, we resort to an amortized analysis by showing that (i) when $v \notin \danger(u)$, $-\beta G$ can not just compensate for $Q$, it can in fact compensate for $\frac{1}{\sqrt{\lambda}}Q \geq  \frac{1}{\sqrt{\lambda}} \cdot \lambda^2 \cosh (\lambda d_{u})$, and secondly, (ii) the probability over a random choice of $v$ of $v \notin \danger(u)$ is at least $\sqrt{\lambda}$, provided $\Phi$ is bounded to begin with. The overall proof then follows from taking an average over all $v$.


Hence, in what follows, we will  show that in expectation the magnitude of $\beta G$ can compensate for a suitably large multiple of $Q$ when $v \notin \danger(u)$. 

\begin{claim} \label{cl:sinh-to-cosh}
Let $\beta \geq 6\lambda$.	For any fixed choice of vertices $u$ and $v$ such that $v
        \notin \danger(u)$, we have $G := \lambda \big( \sinh(\lambda d_{u_>}) - \sinh (\lambda
d_{u_<}) \big) \geq \frac{\lambda}{3} (\cosh(\lambda d_u) + \cosh(\lambda d_v) - 4)$. 
\end{claim}

\begin{proof} The proof is a simple convexity argument.
To this end, suppose both $d_u, d_v \geq
0$. Then since $\sinh(x)$ is convex when $x \geq
0$ and its derivative is $\cosh(x)$, we get that
\begin{align*}
 \sinh(\lambda d_{u_>}) - \sinh(\lambda d_{u_<}) &~~\geq~~ \lambda
  \cosh(\lambda d_{u_<}) \cdot |d_u - d_v| ~~\geq~~ 2 \cosh(\lambda
                                                   d_{u_<}), \\
  \intertext{using
 $v \notin \danger(u)$. But since $\big| |\sinh(x)| -\cosh(x) \big| \leq
  1$, we get that}
  \sinh(\lambda d_{u_>}) - \sinh(\lambda d_{u_<}) &~~\geq~~ 2\sinh(\lambda d_{u_<}) - 2.
\end{align*}
Therefore, $\sinh(\lambda d_{u_<}) \leq \frac13(\sinh(\lambda
d_{u_>})+1)$. Now substituting, and using the monotonicity of $\sinh$ and its closeness to
$\cosh$, we get $G$ is at least
$$  \frac{2 \lambda}{3} \left( \sinh(\lambda d_{u_>}) - 1
\right) ~\geq~ \frac{\lambda}{3} \left( \sinh(\lambda d_{u_>}) +
  \sinh(\lambda d_{u <}) - 2 \right) ~\geq~ \frac{\lambda}{3} \Big(
  \cosh(\lambda d_{u}) + \cosh(\lambda d_{v}) -
  4 \Big).$$ The case of $d_u, d_v \leq 0$ follows from setting $d_u'
= |d_u|, d_v' = |d_v|$ and using the above calculations, keeping in
mind  that
$\sinh$ is an odd function but $\cosh$ is even.  Finally, when $d_{u
  <}$ is negative but $d_{u >}$ is positive, 
\begin{align*}
  G &~~=~~ \lambda( \big( \sinh(\lambda d_{u_>}) - \sinh (\lambda d_{u_<})
  \big) ~~=~~ \lambda \big( \sinh(\lambda d_{u_>}) + \sinh (\lambda
  |d_{u_<}|) \big) \\ 
  &~~\geq~~ \frac\lambda3 \big( \cosh(\lambda d_{u_>}) +
  \cosh (\lambda d_{u_<}) - 2 \big) ~~\geq~~ \frac{\lambda}{3} \Big(
  \cosh(\lambda d_{u}) + \cosh(\lambda d_{v}) -
  4\Big).   \qedhere
\end{align*}
\end{proof}

\begin{claim} \label{cl:no-danger}
Let $\beta \geq 6\lambda$.	For any fixed choice of vertices $u$ and $v$ such that $v
        \notin \danger(u)$, we have $ -\beta G + \left( 1 +
            \frac{1}{\sqrt{\lambda}} \right) Q \leq O(1) $. 
\end{claim}
\begin{proof}
Recall that $G = \lambda \big( \sinh(\lambda d_{u_>}) - \sinh (\lambda
d_{u_<}) \big)$. Now, let $A$ denote $ \cosh(\lambda d_{u}) + \cosh(\lambda
d_{v}).$ Then, by definition of $Q$ and from~\Cref{cl:sinh-to-cosh}, we have that 
$$ -\beta G + \left( 1 + \frac{1}{\sqrt{\lambda}} \right)  Q ~\leq~
-\frac{\beta\lambda}{3} (A - 4) +  \left( 1 + \frac{1}{\sqrt{\lambda}} \right)
\lambda^2 A ~\leq~ \frac{4\lambda \beta}{3} + \left(\lambda^2 +
  \lambda^{\frac{3}{2}} - \frac{\lambda \beta}{3} \right) A ~\leq~
\lambda \beta $$ is at most $O(1)$, assuming $\beta \geq  6 \lambda \geq 3
(\lambda+\sqrt{\lambda})$, and recalling that $\lambda, \beta$ are at most 1. 
\end{proof}

We now proceed with our proof using two cases:


\medskip \noindent {\bf Case (i):} $|d_u| \leq \frac{10}{\lambda}$. In
this case, note that the $Q$ term is
\begin{align*}
  &\E_v[Q \mid u \textrm{ is sampled first}] \\
  & =  \E_v[Q \mid v \in \danger(u), \, u \textrm{ is sampled first}]  \cdot \Pr[v \in \danger(u) \mid u \textrm{ is sampled first}]  \\ 
& ~~~~  + \E_v[Q \mid v \notin \danger(u) \, u \textrm{ is sampled first}] \cdot \Pr[v \notin \danger(u) \mid u \textrm{ is sampled first}] & \\ 
 & \leq  O(1) + \E_v[Q \mid v \notin \danger(u) ~,~ u \textrm{ is sampled first}] \cdot \Pr[v \notin \danger(u) \mid u \textrm{ is sampled first}] .
\end{align*}
Here the inequality uses  $v \in \danger(u)$ and $|d_u| \leq
\frac{10}{\lambda}$ to infer that that both $|d_u|$ and $|d_v|$ are $ \leq
\frac{12}{\lambda}$. Hence the $Q$ term in this scenario will simply be a constant.

Next we analyze the $L$ term. For the following, we observe that the algorithm chooses a random $\pm 1$ signing with probability $(1-\beta)$, and chooses the greedy signing with probability $\beta$, and moreover, this choice is independent of the random choices of $u$ and $v$. Hence, the expected $L$ term conditioned on the algorithm choosing a random signing is simply $0$, and the expected $L$ term conditioned on the algorithm choosing the greedy signing is simply the term $\E[G]$. Hence, we can conclude that: 
\begin{align*}
  &\E_v[-L \mid u \textrm{ is sampled first}] \\
  &=   \E_v[-L \mid v \in \danger(u),  u \textrm{ is sampled first}] \cdot \Pr[v \in \danger(u) \mid u \textrm{ is sampled first}]  \\ 
& ~~~~ + \E_v[-L \mid v \notin \danger(u) ~,~ u \textrm{ is sampled first}] \cdot \Pr[v \notin \danger(u) \mid u \textrm{ is sampled first}] & \\ 
& \leq  \E_v[-\beta G \mid v \notin \danger(u) ~,~ u \textrm{ is sampled first}] \cdot \Pr[v \notin \danger(u) \mid u \textrm{ is sampled first}] .
\end{align*}
Adding the inequalities and applying~\Cref{cl:no-danger}, we get  
$
\E_v[-L + Q \, | \, u \textrm{ is sampled first}] \leq  O(1).$

\medskip \noindent {\bf Case (ii):} $|d_u| > \frac{10}{\lambda}$. 
We first prove two easy claims.
\begin{claim}
\label{cl:easy}
Suppose $v \in \danger(u).$ Then $\cosh(\lambda d_v) \leq 8 \cosh(\lambda d_u). $
\end{claim}
\begin{proof}
  Assume w.l.o.g.\ that $d_u, d_v \geq 0$. Also, assume that
  $d_v \geq d_u,$ otherwise there is nothing to prove. Now 
  $d_v \leq d_u + \frac{2}{\lambda}.$ So
  $\frac{\cosh(\lambda d_v)}{\cosh(\lambda d_u)} \leq \sup_x
  \frac{\cosh(x+2)}{\cosh(x)}$. The supremum on the right happens
  when $x \to \infty$, and then the ratio approaches $e^2 < 8$.
\end{proof}

\begin{claim} \label{cl:good-prob}
    For any discrepancy vector $\bdt$ such that $\Phi(\bdt) \leq O((nT)^{10})$, and for any $u$ such that $|d_u| > \frac{10}{\lambda}$, we have $\Pr[v \notin \danger(u)] \geq 8 \sqrt{\lambda}$, as long as $\lambda = O(\log^{-4} nT)$.
\end{claim}
\begin{proof}
      We consider the case that $d_u > \frac{10}{\lambda}$; the 
    case were $d_u < - \frac{10}{\lambda}$ is similar. 
    
  Assume for a contradiction that
  $\Pr[v \in \danger(u)] \geq 1- 8 \sqrt{\lambda}$, and so $\Pr[v \notin \danger(u)] \leq 8 \sqrt{\lambda}$. 
  We first show that the cardinality of the set $| w \notin \danger(u) | $ is small. Indeed, this follows immediately from our assumption on the minimum weight of any vertex in the statement of~\Cref{lem:beta-process} being at least $\gamma/n$ times the total weight. So we have that for every $w$, the probability of sampling $w$ in the \OPBP is at least $\pi_w \geq \gamma/n$, implying that the total number of vertices not in $\danger(u)$ must be at most $ \frac{8\sqrt{\lambda} \cdot n}{\gamma}$. 
  This also means that the total number of vertices in $\danger(u) \geq \frac{n}{2}$ since $\gamma \geq {\lambda}^{1/4} \geq 16 \sqrt{\lambda}$ for sufficiently small $\lambda$.

Since $d_u > \frac{10}{\lambda}$, we get that any vertex $v \in \danger(u)$ satisfies
  $d_v \geq d_u - \frac2\lambda \geq \frac{8}{\lambda}$.  Moreover, since
  $\sum_{v} d_v = 0$, it must be that the negative discrepancies must in total compensate for the total sum of discrepancies of the vertices in $\danger(u)$. Hence, we have that 
$  \sum_{w : d_w < 0} |d_w| ~~\geq~~ \sum_{v \in \danger(u)}d_v ~~\geq~~  |\{ v~:~ v \in \danger(u)\}|  \cdot \frac{8}{\lambda} ~~\geq~~ 0.5n \cdot \frac{8}{\lambda}$.

From the last inequality, and since $| \{w : d_w < 0 \}| \leq |\{w ~:~ w \not\in \danger(u)\}|
\leq \frac{ 8 \sqrt{\lambda} n}{\gamma}$, we get that there exists a vertex $\widetilde{w}$ s.t $d_{\widetilde{w}} < 0$ and $| d_{\widetilde{w}} | \geq \frac{\gamma}{8 \sqrt{\lambda}
  n} \cdot  \frac{4n}{\lambda} = \frac{\gamma}{2 \lambda^{3/2}} $. But this
implies $\Phi(\bd_t) \geq \cosh(\lambda d_{\widetilde{w}}) \geq \cosh
\left( \frac{\gamma}{2 \sqrt{\lambda}} \right) > (nT)^{10}$, using that
$\lambda = O(\log^{-4} nT)$ and that $\gamma \geq \lambda^{1/4}$. So we get a contradiction on the assumption that $\Phi(\bdt) \leq (nT)^{10}$. 
\end{proof}

Returning to the proof for the case of $|d_u| \geq
\frac{10}{\lambda}$, we get that
\begin{align*}
  &\E_v[Q \mid u \textrm{ is sampled first}] \\
       &=   \E_v[Q \mid v \in \danger(u) ~,~ u \textrm{ is sampled first}] \cdot \Pr[v \in \danger(u) \mid u \textrm{ is sampled first}]  \\ 
& \quad + \E_v[Q \mid v \notin \danger(u) ~,~ u \textrm{ is sampled first}] \cdot \Pr[v \notin \danger(u) \mid u \textrm{ is sampled first}] & \\ 
& \leq  8\lambda^2 \cosh(\lambda d_u) \\
& \quad + \E[Q \mid v \notin \danger(u) ~,~ u \textrm{ is sampled first}] \cdot \Pr[v \notin \danger(u) \mid u \textrm{ is sampled first}] ,
\end{align*}
where the first term in inequality follows from Claim~\ref{cl:easy}. 

Next we analyze the $L$ term similarly:
\begin{align*}
&\E_v[-L \mid \, u \textrm{ is sampled first}]\\
&=   \E_v[-L \mid v \in \danger(u), \, u \textrm{ is sampled first}] \cdot \Pr[v \in \danger(u) \, u \textrm{ is sampled first}]  \\ 
& \qquad  + \E_v[-L \mid v \notin \danger(u) ~,~ u \textrm{ is sampled first}] \cdot \Pr[v \notin \danger(u) \, u \textrm{ is sampled first}] & \\ 
& \leq   \E_v[-\beta G \mid v \notin \danger(u) ~,~ u \textrm{ is sampled first}] \cdot \Pr[v \notin \danger(u) \mid u \textrm{ is sampled first}],
\end{align*}
where the last inequality follows using the same arguments as in case~(i). 
Adding these inequalities and applying~\Cref{cl:no-danger}, we get that 
\begin{align*}
\E_v[-L + Q \mid u \textrm{ is sampled first}] &~~\leq~~ O(1) + 8 \lambda^2 \cosh(\lambda d_u) \\
&\hspace{-3cm} - \frac{1}{\sqrt{\lambda}} \cdot \E_v[ Q \mid  u \textrm{ is sampled first}] \cdot \Pr[v \notin \danger(u) \mid u \textrm{ is sampled first}].
\end{align*}
To complete the proof  of \Cref{lem:beta-process}, we note that $Q \geq \lambda^2 \cosh(\lambda
d_u)$, and use Claim~\ref{cl:good-prob} to infer that $\Pr[v \notin \danger(u)] \geq 8 \sqrt{\lambda}$. This implies
\[ \E_v[-L + Q \mid u \textrm{ is sampled first}] ~\leq~ O(1) + 8 \lambda^2 \cosh(\lambda d_u) 
- 8 \lambda^2 \cosh(\lambda d_u) ~\leq~ O(1). \qedhere \]
\end{proof}





We now can use this one-step expected potential change for the \OPBP
to get the following result for the original expander process: 

\begin{proof}[Proof of~\Cref{thm:main}]
Combining \Cref{lem:beta-process} and \Cref{lem:coupling}, we get that
in the expander process, if we condition on the random choices made
until time $t$, if $\Phi(\bdt) \leq (nT)^{10}$, then $\E[ \Phi(\bd^{t+1}) - \Phi(\bdt)] \leq C$ for some constant $C$.
  The potential starts off at
  $n$, so if it ever exceeds $C\,T\,(nT)^5$ in
  $T$ steps, there must be a time $t$ such that $\Phi(\bd^t) \leq
  C\,t\,(nT)^5$ and the increase is at least
  $C(nT)^5$. But the expected increase at this step is at most
  $C$, so by Markov's inequality the probability of increasing by
  $C(nT)^5$ is at most
  $1/(nT)^5$.  Now a union bound over all times
  $t$ gives that the potential exceeds
  $C\,T\,(nT)^5 \leq (nT)^{10}$ with probability at most $T/(nT)^5 = 1/\poly(nT)$.
  But then $\cosh(\lambda d^t_v) \leq
    (nT)^{10}$, and therefore $d^t_v \leq O(\lambda \log (nT)^{10}) = O(\log^3 
  nT) $ for all vertices $v$ and time $t$.
\end{proof}

In summary, if the underlying graph is
$\gamma$-weakly-regular for $\gamma \geq \Omega(\log^{-1}
nT)$, and has expansion $\alpha \geq \Omega(\log^{-2} nT)$, the greedy
process maintains a poly-logarithmic discrepancy. 


\subsection{Putting it Together} \label{sec:tieup}
We briefly describe the expander decomposition procedure and  summarize the final algorithm. 

\begin{theorem}[Decomposition into Weakly-Regular Expanders]
  \label{thm:regular-expanders}
  Any graph $G = (V,E)$ can be decomposed into an edge-disjoint union
  of smaller graphs $G_1 \uplus G_2 \ldots \uplus G_k$ such that each
  vertex appears in at most $O(\log^2 n)$ many smaller graphs, and (b)
  each of the smaller subgraphs $G_i$ is a $\frac{\alpha}{4}$-weakly regular
  $\alpha$-expander, where $\alpha = O(1/\log n)$.
\end{theorem}

The proof is in~\Cref{sec:expand-decomp}. So, given a graph $G=(V,E),$ we use Theorem~\ref{thm:regular-expanders} to partition the edges into a union of 
$\frac{\alpha}{4}$-weakly regular $\alpha$-expanders, namely $H_1, \ldots, H_s,$ where $\alpha= O(1/\log n)$. Further, each vertex in $V$ appears in at most $O(\log^2 n)$ of these expanders. For each graph $H_i$, we run the greedy algorithm independently. More formally, when an edge $e$ arrives, it belongs to exactly one of the subgraphs $H_i$. We orient this edge with respect to the greedy algorithm running on $H_i$.~\Cref{thm:main} shows that  the discrepancy of each vertex in $H_i$ remains $O(\log^5 (nT))$ for each time $t \in [0 \ldots T]$ with high probability. Since each vertex in $G$ appears in at most $O(\log^2 n)$ such expanders, it follows that the discrepancy of any vertex in $G$ remains $O(\log^7 n + \log^5 T)$ with high probability. This proves~\Cref{thm:final}.



\section{Expander Decomposition}
\label{sec:expand-decomp}

Finally, in this section, we show how to decompose any graph into an
edge-disjoint union of weakly-regular expanders such that no vertex
appears in more than $O(\log^2 n)$ such expanders. Hence, running the
algorithm of the previous section on all these expanders independently
means that the discrepancy of any vertex is at most $O(\log^2 n)$ times
the bound from \Cref{thm:main}, which is $O(\poly\log nT)$ as claimed.
The expander decomposition of this section is not new: it follows from
\cite[Theorem~5.6]{BBGNSSS}, for instance. We give it here for
the sake of completeness, and to explicitly show the bound on the number of
expanders containing any particular vertex.

Recall from \S\ref{sec:notation} that a $\gamma$-weakly-regular
$\alpha$-expander $G = (V,E)$ with $m := |E|$ edges and $n := |V|$
vertices is one where (a) the minimum degree is at least $\gamma$
times the average degree $d_{\rm avg} = \frac{2 m}{n}$, and (b) for
every partition of $V$ into $(S, V \setminus S)$, we have that
$| E(S, V \setminus S) | \geq \alpha \min (\vol(S), \vol(V\setminus
S))$. The main result of this section is the following:

\subsection{Proof of~\Cref{thm:regular-expanders}}

We begin our proof with a definition of what we refer to as \emph{uniformly-dense} graphs.
\begin{defn}[Uniformly Dense Graphs]
	A graph $H = (V,E)$ is \emph{$\alpha$-uniformly-dense} if (i) the minimum degree of the graph $H$ is at least $1/\alpha$ times its average degree $\frac{2m}{n}$, and (ii) no induced sugraph is much denser than $H$, i.e., for every subset $S \subseteq V$, the average degree of the induced sub-graph $\frac{2 E(S,S)}{|S|}$ is at most $\alpha$ times the average degree of $H$ which is $\frac{2m}{n}$. 
\end{defn}

We first provide a procedure which will partition a graph $G$ into edge-disjoint smaller graphs such that each of the smaller graphs is uniformly-dense, and moreoever each vertex participates in $O(\log n)$ such smaller graphs. We then apply a standard expander decomposition on each of the smaller graphs to get our overall decomposition.

\begin{lemma}[Reduction to Uniformly-Dense Instances] \label{lem:uniformly-dense}
	Given any graph $G = (V,E)$, we can decompose it into an edge-disjoint union of smaller graphs $G_1 \uplus G_2 \ldots \uplus G_{\ell}$  such that each vertex appears in at most $O(\log n)$ many smaller graphs, and (b) each of the smaller subgraphs is $2$-uniformly-dense.
\end{lemma}

\begin{proof}
	The following algorithm describes our peeling-off procedure which gives us the desired decomposition.
	
	\begin{algorithm}
		\caption{Input: Graph $G = (V,E)$}
		\label{alg:peeling-off}
		\begin{algorithmic}[1]
			\State initialize the output collection $\cop := \emptyset$. \label{alg:peeling-step1}
			\For {$\bar{d} \in \{ \frac{n}{2}, \frac{n}{4}, \ldots, 32 \}$ in decreasing order}  \label{alg:peeling-step1a}
			\State define the residual graph $R := (V, E_R)$, where $E_R = E \setminus \cup_{G_i = (V_i, E_i) \in \cop} E_i$ is the set of residual edges. \label{alg:peeling-step2}
			\While {there exists vertex $v \in R$ such that $0 < d_R(v) < \bar{d}$} \label{alg:peeling-step3}
			\State delete all edges incident to $v$ from $R$ making $v$ an isolated component. \label{alg:peeling-step4}
			\EndWhile
			\State add each non-trivial connected component in $R$ to $\cop$.  \label{alg:peeling-step5}
			\EndFor
		\end{algorithmic}
	\end{algorithm}

It is easy to see that in any iteration (step~\ref{alg:peeling-step1a}) with degree threshold $\bar{d}$, if a sub-graph $G_i = (V_i, E_i)$ is added to $\cop$ in step~\ref{alg:peeling-step5}, it has minimum degree $\bar{d}$. The crux of the proof is in showing that the average degree of $G_i$ (and in fact of any induced sub-graph of $G_i$) is at most $2 \bar{d}$. Intuitively, this is because the peeling algorithm would have already removed all subgraphs of density more than $2 \bar{d}$ in the previous iterations. We formalize this as follows:

	\begin{claim} \label{cl:peeling-invariant}
		Consider the iteration (step~\ref{alg:peeling-step1a})  when the degree threshold is $\bar{d}$. Then, the residual graph $R$ constructed in step~\ref{alg:peeling-step2} does not have any induced subgraph $S$ of density greater than $2 \bar{d}$.
	\end{claim}
	
	\begin{proof}
		Indeed, for contradiction, suppose there was a subset of vertices in $R$ with average induced degree greater than $2 \bar{d}$. Consider the minimal such subset $S$. Due to the minimality assumption, we in fact get a stronger property that \emph{every vertex in $S$} has induced degree (within $S$) of at least $2 \bar{d}$ (otherwise, we can remove the vertex with minimum induced degree and get a smaller subset $S' \subseteq S$ which still has average induced degree more than $2 \bar{d}$, thereby contradicting the minimality assumption of $S$). 
		
		For ease of notation, let us denote the set of edges induced by $S$ in the graph $R$ as $E_R(S)$. We now claim that all of these edges $E_R(S)$ \emph{should not belong to the residual graph $R$} for this iteration, thereby giving us the desired contradiction. To this end, consider the previous iteration of step~\ref{alg:peeling-step1a} with degree threshold $2 \bar{d}$. Clearly, all of the edges in $E_R(S)$ belong to the residual subgraph for this iteration as well. And consider the first point in the while loop~\ref{alg:peeling-step3} where any edge from $E_R(S)$ is deleted. At this point, note that all the vertices in $S$ must have a degree of strictly greater than $2 \bar{d}$ since even their induced degree in $E_R(S)$ is at least $2 \bar{d}$. Therefore, this gives us an immediate contradiction to any of these edges being deleted in the previous iteration, and hence they would not be present in the current iteration with degree threshold $\bar{d}$. 
	\end{proof}

It is now easy to complete the proof of \Cref{lem:uniformly-dense}. Indeed, we first show that every smaller graph added to $\cop$ in our peeling procedure is $2$-uniformly-dense. To this end, consider any non-trivial connected component added to $\cop$ during some iteration with degree threshold $\bar{d}$. From~\Cref{cl:peeling-invariant}, we know that this component has average degree at most $2 \bar{d}$, and moreover, every vertex in the component has degree at least $\bar{d}$ (otherwise it would be deleted in our while loop). Moreover, every sub-graph induced within this connected component must also have density at most $2 \bar{d}$ again from~\Cref{cl:peeling-invariant}. This then shows that the component added is $2$-uniformly dense.
Finally, 
each vertex participates in at most one non-trivial connected
component in each iteration of step~\ref{alg:peeling-step1a}, and hence
each vertex is present in $O(\log n)$ smaller sub-graphs. Hence the
proof of \Cref{lem:uniformly-dense}.
\end{proof}

Next, we apply a standard divide-and-conquer approach to partition a given 2-uniformly-dense graph $H = (V,E)$ with $m$ edges and $n$ vertices into a vertex-disjoint union of $\alpha$-expanders $H_1 := (V_1, E_1) \uplus H_2 := (V_2, E_2) \ldots \uplus H_k := (V_k, E_k)$, such that the total number of edges in $E$ which are not contained in these expanders is at most $m/2$, and moreover, the induced degree of any vertex in the expander it belongs to is at least $\alpha$ times its degree in $H$.

\begin{lemma}[Decomposition for Uniformly-Dense Graphs] \label{lem:expander-decomp}
	Given any $2$-uniformly-dense graph $H = (V,E)$ with $n$ vertices and $m$ edges, we can decompose the vertex-set $V$ into $V_1 \uplus V_2 \ldots \uplus V_\ell$  such that each induced subgraph $H_i = (V_i, E(V_i))$ is an $\frac{\alpha}{4}$-weakly-regular $\alpha$-expander, and moreover, the total number of edges of $H$ which go between different parts is at most $(2 \alpha \log n ) \, m$. Here $\alpha$ is a parameter which is $O(1/\log n)$. 
\end{lemma}

\begin{proof}
	The following natural recursive algorithm
        (Algorithm~\ref{algo:2}) describes our partitioning
        procedure.\footnote{Step~\ref{alg:decomp-step6} in the
          algorithm does not run in polynomial time. This step can be
          replaced by a suitable logarithmic approximation algorithm,
          which would lose logarithmic terms in the eventual
          discrepancy bound, but would not change the essential nature
          of the result. The details are deferred to the full version.} The only idea which is non-standard is that of using self-loops around vertices during recursion, to capture the property of approximately preserving the degree of every vertex in the final partitioning w.r.t its original degree. This has been applied in other contexts by Thatchaphol et al.~\cite{SW-SODA19}.
	
	\begin{algorithm}[h]
		\caption{Input: Graph $H = (V,E)$}
		\label{alg:sparse-cut}
		\begin{algorithmic}[1]
			\State initialize the output partition $\pop := \emptyset$, and the set of recursive partitions $\rop = \{H := (V,E) \}$. \label{alg:decomp-step1}
			\While {$\rop \neq \emptyset$}  \label{alg:decomp-step2}
			\State choose an arbitrary $H' := (V',E') \in \rop$ to process. \label{alg:decomp-step3}
			\If {the expansion of $H'$ is at least $\alpha$} \label{alg:decomp-step4}
			\State add $H'$ to the final partitioning $\pop$ \label{alg:decomp-step5}
			\Else
			\State let $(S, V' \setminus S)$ denote a cut of conductance at most $\alpha$. \label{alg:decomp-step6}
			\State for each $v \in S$, add $|\delta(v, V' \setminus S)|$ self-loops at $v$. \label{alg:decomp-step7}
			\State for each $v \in V \setminus S$, add $|\delta(v, S)|$ self-loops at $v$. \label{alg:decomp-step8}
			\State add the sub-graphs (including the self-loops) induced in $S$ and $V' \setminus S$ to the recursion set $\rop$ and remove $H'$ from $\rop$. \label{alg:decomp-step9}
			\EndIf
			\EndWhile
		\end{algorithmic}
		\label{algo:2}
	\end{algorithm}
	
\begin{claim} \label{cl:degree-preseved}
Consider any vertex $v$. At all times of the algorithm,  $v$ appears in at most one sub-graph in the collection $\rop$, and morover, suppose it appears in sub-graph $H \in \rop$. Then its degree in $H$ (edges it is incident to plus the number of self-loops it is part of) is exactly its original degree in $G$.
\end{claim}

\begin{proof}
	The proof follows inductively over the number of iterations of the while loop in step~\ref{alg:decomp-step2}. Clearly, at the beginning, $\rop$ contains only $H$, and the claim is satisfied trivially. Suppose it holds until the beginning some iteration $i \geq 1$ of the algorithm. Then during this iteration, two possible scenarios could occur: (a) the algorithm selects a sub-graph $H' \in \rop$, and removes it from $\rop$ and adds it to $\pop$, or (b) the algorithm finds a sparse cut of $H'$ and adds the two induced subgraphs to $\rop$ after removing $H'$ from $\rop$. The inductive claim continues to hold in the first case since we dont add any new graphs to $\rop$. In case (b), note that, for every vertex $v \in H'$, we add as many self-loops as the number of edges incident ot $v$ that cross the partition in the new sub-graph it belongs to. Hence, the  inductive claim holds in this scenario as well.
\end{proof}

\begin{claim} \label{cl:add-expanders}
  Every sub-graph $H'$ which is added to $\pop$ is an $\frac{\alpha}{4}$-weakly-regular $\alpha$-expander.
\end{claim}

\begin{proof}
Consider any iteration of the algorithm where it adds a sub-graph $H'$ to $\pop$ in step~\ref{alg:decomp-step5}. That $H'$ is an $\alpha$-expander is immediate from the condition in step~\ref{alg:decomp-step4}.  Moreover, since the input graph $H$ is $2$-uniformly dense, we know that (a) for every vertex $v \in H$, its degree in $H$ is at least half of the average degree $\bar{d}(H)$ of $H$, and (b) the average degree $\bar{d}(H')$  of $H'$ (which is a sub-graph of $H$) is at most $2 \bar{d}(H)$. Finally, from the fact that $H'$ is an $\alpha$-expander, we can apply the expansion property to each vertex to obtain that $d_{H'}(v) \geq \alpha \cdot \vol_{H'}(v) = \alpha \cdot d_H(v)$. Here, the last equality is due to~\Cref{cl:degree-preseved}. Putting these observations together, we get that for every $v \in H'$, $d_{H'}(v) \geq \alpha \cdot d_H(v) \geq \frac{\alpha}{2} \bar{d}(H) \geq \frac{\alpha}{4} \bar{d}(H')$. This completes the proof.
\end{proof}

\begin{claim} \label{cl:few-edges}
The total number of edges going across different subgraphs in the final partitioning is at most $(2 \alpha \log n ) \, m$. 
\end{claim}
\begin{proof}
The proof proceeds via a standard charging argument. We associate a charge to each vertex which is $0$ initially for all $v \in V$. Then, whenever we separate a sub-graph $H'$ into to smaller sub-graphs $H_1$ and $H_2$ in step~\ref{alg:decomp-step9}, we charge all the crossing edges to the smaller sub-graph $H_1$ as follows: for each $v \in H_1$, we increase its charge by $\alpha \cdot \vol_{H'}(v) = \alpha  \cdot d_{H'}(v) = \alpha \cdot d_H(v)$, where the last equality follows from~\Cref{cl:degree-preseved}. Then it is easy to see that the total number of edges crossing between $H_1$ and $H_2$ is at most the total increase in charge (summed over all vertices in $H_1$) in this iteration (due to the fact that the considered partition is $\alpha$-sparse in $H$). Hence, over all iterations, the total number of edges going across different sub-graphs is at most the total charge summed over all vertices in $V$.

Finally, note that whenever a vertex $v$ is charged a non-zero amount, the sub-graph it belongs to has reduced in size by a factor of at least two, by virtue of our analysis always charging to the smaller sub-graph. Hence, the total charge any vertex $v \in V$ accrues is at most $( \log n \alpha ) d_G(v)$. Summing over all $v \in V$ then completes the proof.
\end{proof}

This completes the proof of \Cref{lem:expander-decomp}.
\end{proof}

We now complete the 
proof of \Cref{thm:regular-expanders}. We first apply~\Cref{lem:uniformly-dense} to partition the input graph $G$ into $O(\log n)$ edge disjoint subgraphs, say, $H_1, \ldots, H_s$, where each vertex of $G$ appears in at most $O(\log n)$ such subgraphs. For each of these sub-graphs $H_i$, we apply~\Cref{lem:expander-decomp} to obtain   $\frac{\alpha}{4}$-weakly-regular $\alpha$-expanders. Across all these partitions, the total number of edges excluded (due to going between parts in~\Cref{lem:expander-decomp}) is at most $m/2$. We recursively apply the above process (i.e.,~\Cref{lem:uniformly-dense} followed by~\Cref{lem:expander-decomp}) to the residual subgraph induced by these excluded edges. Thus, we have $O(\log n)$ such recursive steps, and taking the union of the $O(\log n)$ subgraphs constructed in such step proves~\Cref{thm:regular-expanders}.



\medskip
\noindent
\subsection*{Acknowledgments}
We thank Thatchaphol Saranurak for explaining and pointing us to \cite[Theorem~5.6]{BBGNSSS}. The last author would like to thank Navin Goyal for introducing him to~\cite{AANRSW-Journal98}.


{\small
\bibliographystyle{alpha}
\bibliography{bib}
}




\end{document}